\documentclass[12pt]{article}
\usepackage[utf8]{inputenc}
\usepackage{blindtext}
\usepackage{titlesec}
\usepackage{graphicx}
\usepackage{amsmath}
\usepackage{amsthm}
\usepackage{amssymb}
\usepackage{tikz}
\usetikzlibrary{decorations.pathreplacing}
\usepackage{subfigure}
\usetikzlibrary{bayesnet}
\usetikzlibrary{arrows.meta}
\usetikzlibrary{backgrounds}
\usepackage{color}
\usepackage{xcolor}
\usepackage{afterpage}
\usepackage{algorithm} 
\usepackage{algpseudocode} 
\usepackage{multirow}
\usepackage{mathtools}
\usepackage{eurosym}
\usepackage{textcomp}
\usepackage{siunitx}
\usepackage{subfigure}
\usepackage{hyperref}
\hypersetup{
    colorlinks=true,
    linkcolor=blue,
    citecolor = magenta,
    filecolor= brown,      
    urlcolor=cyan,
    pdftitle={Overleaf Example},
    pdfpagemode=FullScreen,
    }
\usepackage[font=small,labelfont=bf]{caption} % Required for specifying captions to tables and figures
\usepackage{natbib}
\setcitestyle{authoryear,open={},close={}}
\usepackage{graphicx} % Allows including images
\usepackage{booktabs}
\usepackage{geometry}

\DeclareMathSymbol{\shortminus}{\mathbin}{AMSa}{"39}

\newtheorem{proposition}{Proposition}[section]
\newtheorem{theorem}{Theorem}[section]

\newtheorem{lemma}{Lemma}[section]
\newtheorem{corollary}{Corollary}[section]

\newcommand{\R}{\mathbb{R}}

\newcommand{\E}{\mathbb{E}}

\newcommand{\PP}{\mathbb{P}}

\newcommand{\argmax}{\text{argmax}}

\newenvironment{keyword}{
  \vspace{1em}
  \noindent\textbf{Keywords:}
}{\vspace{1em}}

\title{\bf Tackling estimation risk in Kelly investing \\ using options} %% 

\usepackage{authblk}
\author[1,3]{Fabrizio~Lillo}
\author[2]{Piero Mazzarisi}
\author[3]{Ioanna-Yvonni Tsaknaki}

\affil[1]{\footnotesize Department of Mathematics, University of Bologna, Piazza di Porta San Donato 5, Bologna 40126, Italy.}
\affil[2]{\footnotesize University of Siena, Banchi di Sotto 55, Siena 53100, Italy.}
\affil[3]{\footnotesize Scuola Normale Superiore, Piazza dei Cavalieri 7, Pisa 56126, Italy.}

\begin{document}
\maketitle

\begin{abstract}
\noindent The Kelly criterion provides a general framework for optimizing the growth rate of an investment portfolio over time by maximizing the expected logarithmic utility of wealth. However, the optimality condition of the Kelly criterion is highly sensitive to accurate estimates of the probabilities and investment payoffs. Estimation risk can lead to greatly suboptimal portfolios. In a simple binomial model, we show that the introduction of a European option in the Kelly framework can be used to construct a class of growth optimal portfolios that are robust to estimation risk.
\end{abstract}


\section{Introduction}\label{seq:intro}
The Kelly criterion, originally introduced by \cite{kelly}, is a betting strategy applied to investment management that aims to maximize the logarithmic growth of wealth over the long term. By allocating capital in proportion to the expected edge over the odds, the criterion theoretically ensures the highest possible compounded return. Its application to investment management involves calculating optimal position sizes based on expected returns and probabilities, making it particularly appealing for quantitative investors seeking long-term capital efficiency, as supported by many papers, see, e.g., \cite{Breiman_1961, Thorp1975PortfolioCA, Thorp_1997, MacLean2004, Ziemba_caia, Ziemba1} to name but a few.

Despite its theoretical elegance, the Kelly criterion faces two main criticisms: large short-term risk and high sensitivity to estimation errors, see \cite{maclean2010long, ziemba2011using, Ziemba2Samuelson}. In particular, the latter, known also as {\it estimation risk} (in turn part of the general {\it distribution model risk}, see, e.g., \cite{Cont_uncertainty_2006,Breuer_Csiszar_2016}), refers to the fact that even small inaccuracies in the estimates of key input parameters such as expected returns, volatility, and probabilities can lead to significant over- or under-investing and suboptimal outcomes. As a result, while the Kelly criterion provides a compelling framework for capital allocation, its practical implementation often necessitates conservative adjustments and careful parameter calibration. Whereas short-term market risk is typically moderated using the so-called {\it fractional} Kelly strategy, see \cite{maclean2010long}, suboptimality due to estimation risk is still an open issue.

Here, in the context of a binomial tree market, we prove that the addition of a European option to the investment assets provides optimal Kelly strategies that are robust to estimation risk. More precisely, in the absence of estimation risk, the inclusion of a derivative does not modify the growth rate of the optimal portfolio. In contrast, the two Kelly strategies, with and without options, perform differently, and neither consistently outperforms the other across all parameters misspecifications. We then demonstrate that a proper convex combination of two Kelly portfolios is robust to estimation risk in the long term. 

\section{Review of the classical Kelly strategy}\label{sec:KS}

Let us consider a market where a stock $S$ and a bond $B$ can be traded, described by a time-discrete stochastic binomial tree model with the stock price evolving as a recombining tree where at each step it can move up by a factor $u$ or down by a factor $d$, such that $0<d<u$, and the bond price follows the deterministic dynamics $B_t= B_{t-1}R$ with $R$ the rate of interest,  satisfying the no-arbitrage condition $d<R<u$. Let $X_t\sim\Phi(p)$ be an i.i.d. Bernoulli random variable describing the total return $S_t/S_{t-1}$ at time $t$, i.e. $\PP(X_t=u)=p$ and $\PP(X_t=d)=1-p$ for some $0<p<1$ such that $S_t = S_{t-1}X_t$. Given $n\in\mathbb{N}$, let $(\Omega, \{\mathcal{F}_t\}_{t=0}^n,\PP)$ be the probability space with $\Omega=\{u,d\}^n$ the sample space, $\PP$ the binomial probability measure over $\Omega$, and $\{\mathcal{F}_t\}_{t=0}^n$ the filtration where $\mathcal{F}_t$ denotes the sigma-algebra generated by the process up to time $t$.

At time $t=0$, let $S_0$ and $B_0$ be the stock and the bond prices, respectively. If $W_0$ indicates the initial wealth, then $N_0^{(s)}=\frac{W_0}{S_0}f$ and $N_0^{(b)}=\frac{W_0}{B_0}(1-f)$ represent the number of stock shares and bonds purchased or sold, respectively.\footnote{No constraints on financing (i.e. negative fraction $1-f$ for the bond) or short selling of stocks are considered here.} As a consequence, at time $t=1$, the value in stocks  is $N_0^{(s)}S_1 = N_0^{(s)}S_0X_1$, while, similarly, the value in bonds is $N_0^{(b)}B_0R$. The portfolio's wealth is then $W_1 = W_0 [fX_1+(1-f)R]$. When the fraction $f$ is constant over time, then the dynamics of the portfolio's wealth at a generic time $t$ is described as
\begin{equation}\label{eq:wealth}
W_t = W_{t-1}[fX_t+(1-f)R] = W_{t-1}\pi_f(X_t)\Rightarrow W_n=W_0\prod_{t=1}^n\pi_f(X_t),
\end{equation}
with $n$ the final time and $\pi_f(X_t):= [fX_t+(1-f)R]$ the relative payoff depending on the random variable $X_t$, that is
\begin{equation}\label{eq:KSpayofft}
{\pi}_{f}(X_t) = \begin{cases}
fu+(1-f) R &\quad\mbox{if}\quad X_t = u,\\
fd+(1-f) R &\quad\mbox{if}\quad X_t = d.
\end{cases}
\end{equation}
By the strong law of large numbers, the long-term exponential growth rate for a generic value of $f$ converges to
\begin{equation}\label{eq:growth}
G_n = \frac{1}{n}\log\frac{W_n}{W_0} = \frac{1}{n}\sum_{t=1}^n\log \pi_f(X_t)\overset{a.s.}\longrightarrow\E[\log \pi_f(X)]\quad\mbox{as }n\rightarrow\infty
\end{equation}
with $X\sim \Phi(p)$ i.i.d. Bernoulli random variable, as long as the relative payoff $ \pi_f(X)$ is positive for any possible outcome. By defining the asymptotic exponential growth rate as $G(f;\Phi) := \E[\log \pi_f(X)]$, it turns out that the Kelly criterion of maximizing $G(f;\Phi)$ is equivalent to the maximization of the log-utility of the (relative) wealth. In other words, the Kelly solution, which we name {\it Kelly Strategy} (KS) throughout the paper, identifies the log-optimal portfolio. 
% The optimal fraction is the solution of the maximization problem $f^*:=f^{*}(\Phi) = \argmax_{f\in\R}G(f;\Phi)$ s.t. $\pi_f(u)> 0,\:\:\pi_f(d)> 0$. 
A closed-form expression is obtained by solving the Karush–Kuhn–Tucker (KKT) conditions. 
That is maximizing the growth rate function $G(f;\Phi)$ subject to
\begin{equation}\label{eq: constraints KS problem}
    \pi_f(X)>0\Leftrightarrow\begin{matrix}
        fu+(1-f)R>0\\[10pt]
        fd+(1-f)R>0
    \end{matrix}\Leftrightarrow\begin{matrix}
        -f-\frac{R}{u-R}<0\\[10pt]
        f-\frac{R}{d-R}<0
    \end{matrix}
\end{equation}
We write the Lagrangian
\begin{equation*}
    \mathcal{L}(f,\lambda_1,\lambda_2) = -G(f;\Phi)+\lambda_1\Big(-f-\frac{R}{u-R}\Big)+\lambda_2\Big(f-\frac{R}{d-R}\Big)
\end{equation*}
since the constraints in Eq.(\ref{eq: constraints KS problem}) are strict we end up in $\lambda_1=\lambda_2=0$, thus the solution of the problem $f^{*}(\Phi) = \argmax_{f\in\R}G(f,\Phi)$, can be solved by using the first order condition, i.e.
\begin{align*}
    G^\prime(f;\Phi) & = 0\Leftrightarrow\\\nonumber
    \frac{p(u-R)}{R+f(u-R)}+\frac{(1-p)(d-R)}{R+f(d-R)} & = 0,
\end{align*}
resulting in
\begin{equation}\label{eq: optimal fraction stocks}
    f^*= \frac{p(R-u)R+(1-p)(R-d)R}{(u-R)(d-R)}\in\R.
\end{equation}
The asymptotic long-term growth rate is
\begin{align}\label{eq: maximum growth rate KS}
\max_{f\in\R}G(f;\Phi)&=G(f^{*};\Phi)=\E[\log \pi_{f^*}(X)] \\&= p\log(f^*(u-R)+R)+(1-p)\log(f^*(d-R)+R)\nonumber
\end{align}
with $\pi_{f^*}(X) = f^*X+(1-f^*)R$. 
The solution in Eq.(\ref{eq: optimal fraction stocks}) is for the unconstrained maximization problem. When financing through bonds, financial leverage, and short selling are not allowed in the market, $f$ should be constrained in the unit interval. The constrained problem can be solved in closed form by solving the KKT conditions.%, see, e.g., \cite{LoOrrZhang2018}.
\begin{proposition}
The {\it constrained} optimal fraction $f^*\in[0,1]$ maximizing $G(f;\Phi)$ is
\begin{equation*}\label{eq: prop in Lo et al}
       f^{*} = \begin{cases}
           \begin{matrix}
               1, & \text{if}\quad\E\big[\frac{X}{R}\big]>1\quad\text{and}\quad \E\big[\frac{R}{X}\big]<1\\[10pt]
               \frac{p(R-u)R+(1-p)(R-d)R}{(u-R)(d-R)}, & \text{if}\quad\E\big[\frac{X}{R}\big]\geq 1\quad\text{and}\quad \E\big[\frac{R}{X}\big]\geq 1\\[10pt]
               0, & \text{if}\quad\E\big[\frac{X}{R}\big]<1\quad\text{and}\quad \E\big[\frac{R}{X}\big]>1
           \end{matrix}
       \end{cases}.
\end{equation*}
\end{proposition}

The proof of the proposition is a direct application of the results in \cite{Brennan_Lo_2011}.

% Suppose to invest at each time a constant fraction $f$ of the wealth in the stock and the remaining in the bond, the wealth at time $t$ is $W_t=W_{t-1}(fX_t+(1-f)R)$. The Kelly criterion (see appendix \ref{sec:KS}) looks for the value of $f$ maximizing the growth rate, defined as
% $G_n=\frac{1}{n}\log\frac{W_n}{W_0}$. For $n \to \infty$ and using the law of large numbers one finds that $G_n\to \mathbb{E}[fX+(1-f)R]$ and the optimal investment fraction $f^*$ is easily obtained in closed form as a function of $R$, $p$, $u$, $d$ 
% \begin{equation}\label{eq: optimal fraction stocks}
%     f^*= \frac{p(R-u)R+(1-p)(R-d)R}{(u-R)(d-R)}\in\R.
% \end{equation}

\subsection{Simulation results for KS}

\begin{figure}[h]
\centering
\includegraphics[width=1\columnwidth]{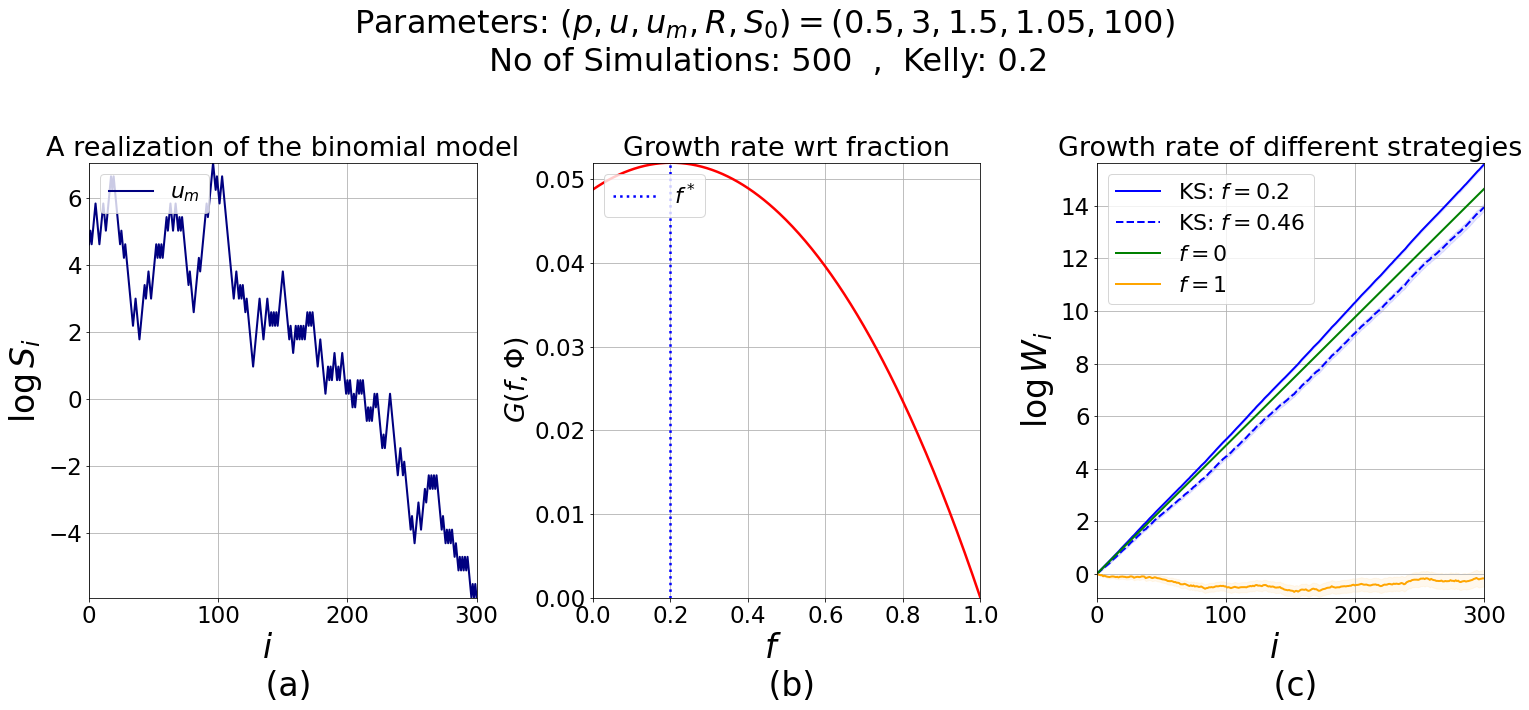}
\caption{\footnotesize \textbf{(a)} A realization of the $\log$-price, with the price evolving as in the binomial model, over 300 rounds with $u_m = 1.5$ in dark blue line. \textbf{(b)} The growth rate function (red line) as a function of the fraction $f$ for the KS strategy when $u=u_m$ and $d=1/u$. The blue dotted line indicates the optimal fraction $f^*=0.2$ (see Eq.(\ref{eq: optimal fraction stocks})). \textbf{(c)} Comparison of the strategies in terms of the log-wealth as a function of the trading rounds. The blue solid line represents the KS strategy in the well-specified scenario when $u=u_m$ and $d=1/u$. The green solid line represents the fixed-income strategy. The yellow solid line represents the strategy that invests everything to the stock. The blue dotted line represents the KS strategy in the presence of misspecification in the parameters, i.e. when $u=3$ and $u_m=1.5$.}
\label{fig: Kelly and stocks}
\end{figure}
We show here simulations of KS to verify numerically that Kelly's solution is associated with the highest growth rate compared to other constant fraction strategies. We set $u_\text{m}=1.5$, $d_\text{m}=1/u_\text{m}$, $p=1/2$, and the interest rate for the bond equal to $R=1.05$.

The dark blue line in Figure \ref{fig: Kelly and stocks}(a) -  shows one of the 500 random walks of the stock price over $n=300$ time steps.   
For KS, given the choice of the parameters, the optimal fraction is $f^*=0.2$. Figure \ref{fig: Kelly and stocks} (b) shows the growth rate  $G(f;\Phi)$ as a function of $f$, indicating also the optimal $f^*$ with a dotted blue line. Finally, in Figure \ref{fig: Kelly and stocks} (c), we compare the performance of Kelly's solution with that of three constant fraction alternative strategies. The performance is the log-wealth as a function of time. The three strategies are: (i) investing entirely in bonds ($f=0$ - green line) - the fixed-income strategy, (ii) investing entirely in stocks ($f=1$ - orange line), and (iii) a sub-optimal KS for $f=0.46$. We perform $N=500$ simulations, averaging the growth rate at each round for each strategy, and include standard errors as shaded regions. In all cases, KS consistently outperforms the alternatives. However, if the market parameters are ill-specified in determining $f^*$, a phenomenon known as {\it estimation risk}, the growth rate of the portfolio could be strongly suboptimal as we show below.

\section{Kelly criterion with options}\label{seq:kelly}

We consider now an extended Kelly framework where part of the wealth can be invested in European options. In particular, we include the possibility of trading one-period European put options with a strike price within the range between the stock’s lower and upper possible values. 
The problem refers to finding constant fractions $f$, $g$, and $1-f-g$ of the wealth to be invested in stock, option, and bond, respectively. We do not consider here any constraints on financing (i.e. negative fraction $1-f-g$ for the bond) or short selling of stocks/options. The constrained version of the problem is a simple extension of the results shown here.

Let $S_0$, $B_0$, $K_0$, and $P_0$ be the price of the stock, the bond, the strike, and the option\footnote{See, for example, \cite{Shreve_book_Binomial} for the derivation of the classical option pricing formula for a binomial model of a put option with strike price $K_0$ and payoff $(K_0-S_1)^{+}=\max\{K_0-S_1,0\}$ at maturity $t=1$, that is
$$
P_0 = \frac{1}{R}\frac{u-R}{u-d}(K_0-dS_0).
$$} at time $t=0$, respectively. If $W_0$ indicates the initial wealth, then $N_0^{(s)}=\frac{W_0}{S_0}f$, $N_0^{(b)}=\frac{W_0}{B_0}(1-f-g)$, and $N_0^{(o)}=\frac{W_0}{P_0}g$ represent the number of stock shares, bonds, and put options purchased or sold, respectively. As a consequence, at time $t=1$, the value in stocks is $N_0^{(s)}S_1 = N_0^{(s)}S_0X_1$, the value in bonds is $N_0^{(b)}B_0R$, while the payoff of the put options is $N_0^{(o)}(K_0-S_0X_1)^+$. As such, the wealth after one period is
\begin{align}
    {W}_1 & = N_0^{(o)}(K_0-S_1)^{+}+N_0^{(s)}S_1+N_0^{(b)}B_0R\nonumber\\
    & = W_0\Big[\frac{g}{P_0}(K_0-S_0X_1)^{+}+f(X_1-R)+(1-g)R\Big]\nonumber\\
    & = W_0\Big\{g\Big[\frac{(K_0-S_0X_1)^+}{P_0}+(X_1-R)\frac{S_0}{P_0}\Big]+(1-g)R+c(X_1-R)\Big\}\label{eq: wealth KO justification}
\end{align}
We divide by $W_0$ and set $c := f-(S_0/P_0)g$ to get the portfolio wealth's relative payoff:
\begin{equation}\label{eq:KOpayoff}
%\tilde{\pi}_{g,c}(X_1) 
\frac{W_1}{W_0}= \begin{cases}
g\left[u\frac{S_0}{P_0}-\frac{S_0}{P_0}R\right]+(1-g)R+c(u-R) &\quad\mbox{if}\quad X_1 = u\\[10pt]
g\left[\frac{K_0}{P_0}-\frac{S_0}{P_0}R\right]+(1-g)R+c(d-R) &\quad\mbox{if}\quad X_1 = d
\end{cases}
\end{equation}
 Defining $N^{(s)}_{res} = W_0/S_0$ as the maximum number of stocks that can be purchased within the restricted scenario of no leverage and no short selling (i.e. when $f\in[0,1]$), one can rewrite $c=(N_0^{(s)}-N_0^{(o)})/N^{(s)}_{res}$ and interpret it as a hedging strategy parameter determining how many put options are used to cover the stock position. In the following, we parametrize the problem in terms of $(g,c)$ instead of portfolio weights $(f,g)$. For example, $g=0$ and $c=f$ is the classical Kelly Strategy (KS) obtained without using options, while 
 $c=0$ means that the portfolio is composed of the same number of stock shares and put options. As a final remark, we notice that the relative payoff in Eq.~(\ref{eq:KOpayoff}) depends on the strike price $K_0$.
 
% In particular, when $g=0$, we obtain the relative payoff $\pi_{f}(X_1)\equiv \tilde{\pi}_{0,f}(X_1)$ of the standard Kelly portfolio.

 Moving into the multi-period setting, in order to recover the standard Kelly framework for the multi-period binomial tree, namely a ``sequential betting'' characterized by constant fractions and payoff  odds over time, we need  to impose 
 \begin{equation}\label{eq: condition of stationarity}
     \frac{K_t}{P_t}=\frac{K_0}{P_0}~~~\text{for any }t.
 \end{equation}
Conditioning on $\mathcal{F}_t$, or, in other words, given the number $m$ of up price movements until time $t$, the previous condition corresponds, after some algebra, to 
 \begin{equation*}
     K_t = K_0 d^{t-m}u^m\Leftrightarrow \log K_t = \log K_0+(2m-t)\log u
 \end{equation*}
 for any $t$.
Using the condition in Eq.~(\ref{eq: condition of stationarity}), it is easy to show that the ratio $\frac{S_t}{K_t}=\frac{S_0}{K_0}$ is constant over time (see Figure \ref{fig: tikz-strike} for a pictorial illustration of the resulting Kelly strategy), and $\frac{S_t}{P_t} = \frac{S_0}{P_0}$ for any $t$ as a consequence.
% This condition implies that the relative strike price remains constant over time, as it is illustrated in Figure \ref{fig: tikz-strike}, hence the log-strike price is at a constant distance from the log-price across time.
\begin{figure}[h!]
    \begin{center}
\begin{tikzpicture}[scale=0.52]
    [%%%%%%%%%%%%%%%%%%%%%%%%%%%%%%
        dot/.style={circle,draw=black, fill,inner sep=1pt},
    ]%%%%%%%%%%%%%%%%%%%%%%%%%%%%%%
\node[black,circle,fill,inner sep=1.1pt] at (4,0){}; 
\node[black,circle,fill,inner sep=1.1pt] at (5,1){};
\node[black,circle,fill,inner sep=1.1pt] at (3,1){};
\node[black,circle,fill,inner sep=1.1pt] at (6,2){};
\node[black,circle,fill,inner sep=1.1pt] at (4,2){};
\node[black,circle,fill,inner sep=1.1pt] at (2,2){};
\node[black,circle,fill,inner sep=1.1pt] at (7,3){};
\node[black,circle,fill,inner sep=1.1pt] at (3,3){};
\node[black,circle,fill,inner sep=1.1pt] at (5,3){};
\node[black,circle,fill,inner sep=1.1pt] at (1,3){};
\node[black,circle,fill,inner sep=1.1pt] at (8,4){};
\node[black,circle,fill,inner sep=1.1pt] at (6,4){};
\node[black,circle,fill,inner sep=1.1pt] at (4,4){};
\node[black,circle,fill,inner sep=1.1pt] at (2,4){};
\node[black,circle,fill,inner sep=1.1pt] at (0,4){};
\node[blue,circle,fill,inner sep=1.1pt] at (3.4,1){};
\node[blue,circle,fill,inner sep=1.1pt] at (2.4,2){};
\node[blue,circle,fill,inner sep=1.1pt] at (4.4,2){};
\node[blue,circle,fill,inner sep=1.1pt] at (3.4,3){};
\node[blue,circle,fill,inner sep=1.1pt] at (1.4,3){};
\node[blue,circle,fill,inner sep=1.1pt] at (5.4,3){};
\node[blue,circle,fill,inner sep=1.1pt] at (0.4,4){};
\node[blue,circle,fill,inner sep=1.1pt] at (2.4,4){};
\node[blue,circle,fill,inner sep=1.1pt] at (4.4,4){};
\node[blue,circle,fill,inner sep=1.1pt] at (6.4,4){};

\draw[-,gray,dashed] (4,0) -- (8,4);
\draw[-,gray,dashed] (4,0) -- (0,4);
\draw[-,gray,dashed] (3,1) -- (6,4);
\draw[-,blue,dotted] (3.4,1) -- (6.4,4);
\draw[-,gray,dashed] (5,1) -- (2,4);
\draw[-,blue,dotted] (5.4,3) -- (4.4,4);
\draw[-,blue,dotted] (4.4,2) -- (2.4,4);
\draw[-,blue,dotted] (3.4,1) -- (0.4,4);
\draw[-,gray,dashed] (6,2) -- (4,4);
\draw[-,gray,dashed] (2,2) -- (4,4);
\draw[-,blue,dotted] (2.4,2) -- (4.4,4);
\draw[-,gray,dashed] (7,3) -- (6,4);
\draw[-,gray,dashed] (1,3) -- (2,4);
\draw[-,blue,dotted] (1.4,3) -- (2.4,4);

\draw[-,dotted] (-1,1) -- (9.2,1);
\draw[-,dotted] (-1,2) -- (9.2,2);
\draw[-,dotted] (-1,3) -- (9.2,3);
\draw[-,dotted] (-1,4) -- (9.2,4);

\foreach \y in {1,...,2}
    \draw (-1,\y) -- node[xshift=-2mm] {\y} (-1.1,\y); 
\draw (-1,4) -- node[xshift=-2mm] {$n$} (-1.1,4); 

\draw (4,0) -- node[below,xshift=.3mm] {$\log S_0$} (4,-.1);
\draw[blue] (3.4,0) -- node[below,xshift=-.1mm,yshift=-1.3mm] {$\log K_0$} (3.4,-1.3); 
\node[above] at (-1.4,5){$i$};
\node[above] at (8.8,0.2){$\log s$};
    
\draw[->,thin,gray,-latex] (-1,-1) -- (-1,5.5);
\draw[->,thin,gray,-latex] (-2,0) -- (9,0);
\end{tikzpicture}
\end{center}
    \caption{An illustration of the log-strike price (blue bullets) over time on the binomial tree for stock prices (black bullets).}
    \label{fig: tikz-strike}
\end{figure}
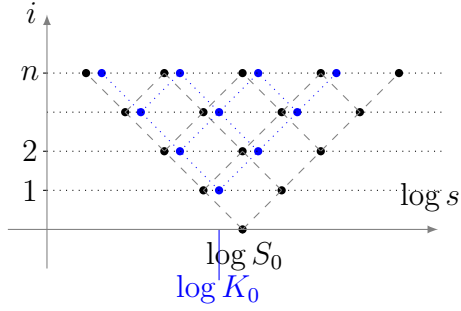
%Consequently, at each time step, the KO strategy essentially involves selecting two parameters: the strike price and the choice of $g$. Then, $\frac{S_t}{P_t} = \frac{S_0}{P_0}$ for any $t$ as well. 
Since a Kelly strategy considers constant fractions over time, the previous condition implies also that $c_t = f-(S_t/P_t)g = c$ for any $t$. Finally, the odds in Eq.~(\ref{eq:KOpayoff}) are constant over time as well when fractions $f$ and $g$ are constant, under the previous condition on the strike price. It is
\begin{equation}\label{eq:KOpayofft}
%\tilde{\pi}_{g,c}(X_t) 
\pi_{g,c}(X_t)\equiv\frac{W_t}{W_{t-1}}
= \begin{cases}
g\left[u\frac{S_t}{P_t}-\frac{S_t}{P_t}R\right]+(1-g)R+c(u-R) &\quad\mbox{if}\quad X_t = u\\[10pt]
g\left[\frac{K_t}{P_t}-\frac{S_t}{P_t}R\right]+(1-g)R+c(d-R) &\quad\mbox{if}\quad X_t = d
\end{cases}
\end{equation}
and the law of $\pi_{g,c}(X_t)$ is independent from $t$. The wealth at time $n$ is ${W}_n = W_0\prod_{t=1}^n W_t/W_{t-1}$, and by the strong law of large numbers, the long-term exponential growth rate converges as $n\rightarrow\infty$ to
\begin{equation}\label{eq:KOgrowth}
{G}_n(g,c;\Phi) = \frac{1}{n}\log\frac{{W}_n}{W_0} = \frac{1}{n} \sum_{t=1}^n \log {\pi_{g,c}}(X_t)\overset{a.s.}\longrightarrow \E[\log{\pi_{g,c}}(X)] \equiv G(g,c;\Phi)
\end{equation}
for positive relative payoff ${\pi_{g,c}}(X)$. The Kelly criterion prescribes to maximize ${G}(g,c;\Phi)$ and we refer to the solution as the {\it Kelly with Option} (KO) strategy. For a specific hedging strategy defined by a value of $c$, the optimal fraction invested in options is $g^* = \mbox{argmax}_{g\in\mathbb{R}}{G}(g,c;\Phi)$ s.t. ${\pi}_{g,c}(u)>0$, ${\pi}_{g,c}(d)>0$. The optimal $g^*$ is the solution of the KKT conditions obtained by maximizing the growth rate
\begin{equation*}
   {G}(g,c;\Phi) = p\log{\pi}_{g,c}(u)+(1-p)\log{\pi}_{g,c}(d)
\end{equation*}
subject to
\begin{equation}\label{eq: strict inequalities}
    \begin{matrix}
        {\pi}_{g,c}(u)>0\\[10pt]
        {\pi}_{g,c}(d)>0
    \end{matrix}\Leftrightarrow \begin{matrix}
        -c\frac{u-2R}{\tilde{u}-R}-g<0\\[10pt]
        -c\frac{d-2R}{\tilde{d}-R}-g<0
    \end{matrix}
\end{equation}
where
\begin{equation}
    \tilde{u} = u\frac{S_0}{P_0}-R\frac{S_0}{P_0}\quad\text{and}\quad\tilde{d} = \frac{K_0}{P_0}-R\frac{S_0}{P_0}.
\end{equation}
Let us define the Lagrangian as
\begin{equation*}
    \mathcal{L}(g,\lambda_1,\lambda_2) = -{G}(g,c;\Phi)+\lambda_1\Big(-c\frac{u-2R}{\tilde{u}-R}-g\Big)+\lambda_2\Big(-c\frac{d-2R}{\tilde{d}-R}-g\Big).
\end{equation*}
Due to the strict inequalities in Eq.~(\ref{eq: strict inequalities}), the Lagrange multipliers  $\lambda_1,\lambda_2$ must be zero. Hence, $g^*$ is the solution of the first-order condition 
\begin{align}
    \frac{\partial{G}}{\partial g} & = 0\Leftrightarrow\nonumber\\\label{eq: optimal solution KO}
    g^*&=\frac{p(R-\tilde{u})(R+c (d-R))+(1-p)(R-\tilde{d})(R+c(u-R))}{(\tilde{d}-R)(\tilde{u}-R)}.
    \end{align}
We also observe that $g^*$ is a linear function of $c$
\begin{equation*}
    g^* = -c\frac{u-R}{\tilde{u}-R}-p\frac{R}{\tilde{d}-R}-(1-p)\frac{R}{\tilde{u}-R}.
\end{equation*}

\begin{figure}[t]
    \centering
     \subfigure[]{\includegraphics[width=0.49\textwidth]{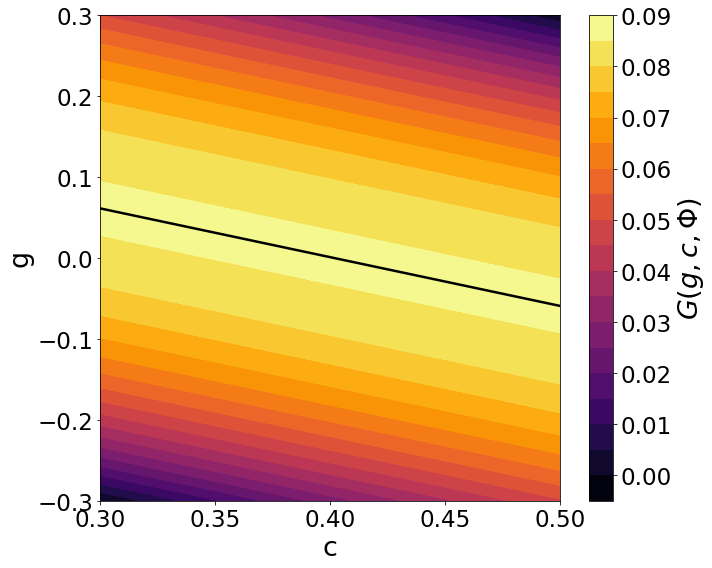}}  
      \subfigure[]{\includegraphics[width=0.49\textwidth]{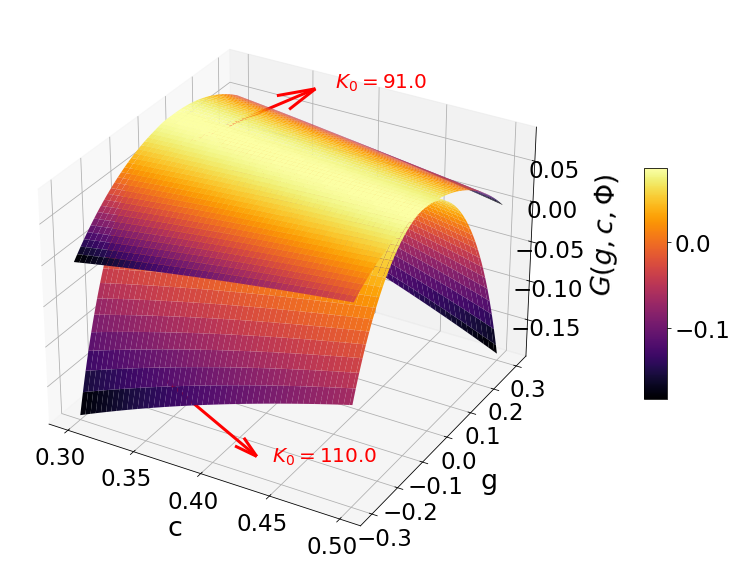}} 

    \caption{\footnotesize (a) Contour plot of the asymptotic growth rate in Eq.~(\ref{eq:KOgrowth}) for $S_0=100$, $K_0=110$, $u=2$, $d=1/u$, $p=0.5$, and $R=1.05$. The black line corresponds to the optimal solutions. (b) Growth rate surface for $K_0=91$ and $K_0=110$.}
    \label{fig:Gko}
\end{figure}
%The following proposition shows that the KS is a replicating portfolio of the KO strategy in the absence of estimation risk.
%\begin{proposition}\label{prop: replicating KO for any c}
%Let $K_0\in(dS_0,uS_0)$ and assume that the option is priced arbitrage-free then the optimal KS replicates the optimal KO strategy for any $c\in\R$.
%\end{proposition}
%\begin{proof}
%    See appendix \ref{app:proofs}.
%\end{proof}
%The following proposition is an immediate corollary of the previous proposition, however it can also be proven independently from that, thus we provide an alternative proof in the appendix.

The following proposition characterizes the KO solution with respect to the classical KS in the case of no estimation risk.

\begin{proposition}\label{prop:1}
Let $(\Omega, \{\mathcal{F}_t\}_{t=0}^n,\PP)$ be the probability space associated with the binomial tree market, and $X\sim\Phi(p)$ an i.i.d. Bernoulli variable associated with the price dynamics. Let $K_0\in(dS_0,uS_0)$ then, for any $c\in\mathbb{R}$, it is
\begin{equation}\label{eq:equ}
\pi_{g^*,c}(X) = \pi_{0,f^*}(X) \quad\mbox{a.s.}
\end{equation}
where $g^*$ and $f^*$ are the optimal fractions solving the Kelly criterion for KO and KS, respectively.
\end{proposition}
\begin{proof}
    See appendix \ref{app:proofs}.
\end{proof}

Proposition~\ref{prop:1} tells that the relative payoff of the optimal KO strategy is the same as the one of the standard KS strategy for any hedging strategy and any strike price. As a consequence, the asymptotic growth rate of the two strategies is also the same. This result is consistent with the absence of arbitrage. Moreover Proposition~\ref{prop:1} is equivalent with the following 
\begin{proposition}\label{prop: replicating KO for any c}
Let $K_0\in(dS_0,uS_0)$ and assume that the option is priced arbitrage-free, then the (unique) optimal KS replicates the optimal KO strategy for any $c\in\R$.
\end{proposition}
In other words, the KS is the unique replicating portfolio of any optimal KO.
If the relative payoff of the optimal KO strategy differs from that of the optimal KS strategy for some value of $c$, it would imply the possibility of earning an additional profit using options. However, this is untenable, as it would mean an arbitrage opportunity, whereas the option is priced arbitrage-free.

The left panel of Figure \ref{fig:Gko} shows the contour plot of the asymptotic growth rate in Eq.~(\ref{eq:KOgrowth}) for a specific strike. The black line indicates the set of equivalent optimal solutions, which have the same portfolio growth rate. The right panel shows the contour plot of the asymptotic growth rate for two different strikes. Although the two surfaces are different, they coincide at their maximum, showing that the optimal growth rate of the KO strategy does not depend on the strike price (as predicted by Proposition~\ref{prop:1} in the following section). 

In conclusion, when the parameters of the Kelly strategies are well-specified, purchasing options offers no advantage since both strategies coincide, or, in other words, the optimal KS is the replicating portfolio of the KO solution. On the contrary, when parameters are ill-specified, namely in the presence of estimation risk, the two strategies are not equivalent anymore (as it can be shown using simple algebra), but options can be used to hedge against possible misspecifications about price dynamics, as follows.

%======================================
\section{Managing estimation risk in Kelly investing}\label{sec:ER}

The above discussion assumes that the parameters $u$, $d$, and $p$, representing investors' beliefs/estimates on market dynamics and used {\it ex-ante} in creating the portfolios, are the same as the actual one. This is unrealistic, since they are often estimated noisily from data and might be also subject to non-stationarity. To see the effect of estimation risk in Kelly investing, let us assume that the portfolio is constructed using a value of $u$ which is different\footnote{This can, for example, be interpreted as a misestimation of volatility.} from the market's realized one $u_\text{m}$.
\iffalse
Within this context, estimation risk refers to a mismatch between the {\it expected} price relatives $u$ and $d=1/u$ (used in the Kelly optimization problem and for option pricing), representing the market belief on returns' distribution, and the {\it realized} price relatives $u_{ex \shortminus post}$ and $d_{market}= 1/u_{ex \shortminus post}$, determining the {\it true} payoffs. The underlying intuition is to interpret the binomial tree with parameterization $u={\exp(\sigma\sqrt{\Delta t})}$ and $d=\exp(-\sigma\sqrt{\Delta t})$ ($\Delta t=1$ in the results before) as the discretization of the Geometric Brownian motion, which can be recovered in the limit $\Delta t \rightarrow 0$ as shown in \cite{cox1979option}. As such, a misspecification of the parameter $u$ can be interpreted as a misvaluation or estimation error of volatility.

From Proposition~\ref{prop:1}, it is clear that the Kelly optimality condition implies $G(f^*,\Phi) = \tilde{G}(g^*,c,\Phi)$ with $f^*\equiv f^*(u)$ and $g^*\equiv g^*(u)$, for any $c$. 
%From explicit solutions for the optimal fraction of KS and KO in Eqs.~(\ref{eq: optimal fraction stocks}) and (\ref{eq: optimal solution KO}), it is evident that $f^* \equiv f^*(u)$ and $g^*\equiv g^*(u)$, with $u$ the {\it expected} price relative. However, we assume here that the market will realize with a different $u_{ex \shortminus post}$, thus the odds in Eqs.~(\ref{eq:KSpayofft}) and (\ref{eq:KOpayofft}) will depend on $u_{ex \shortminus post}$ and $d_{market} = 1/u_{ex \shortminus post}$.
\fi
\begin{figure}[t]
    \centering
    \subfigure[]{\includegraphics[width=0.3\textwidth]{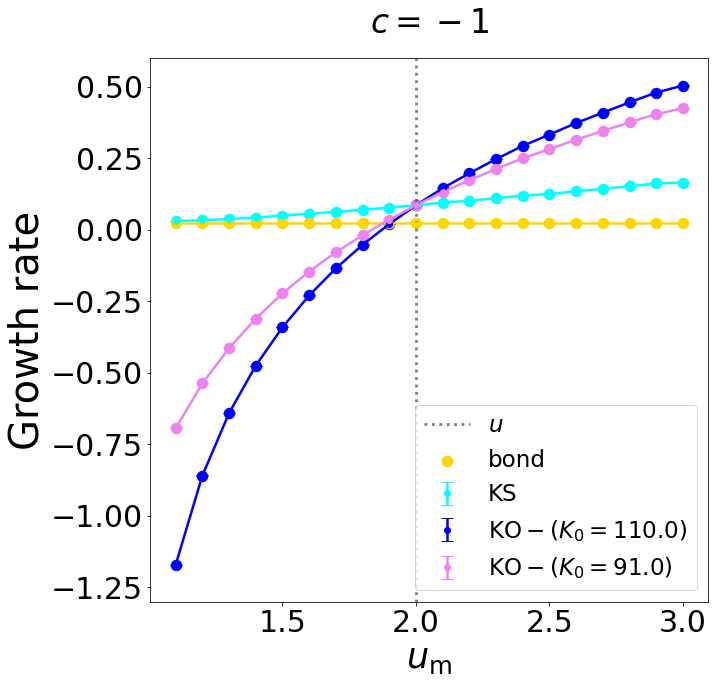}}  
      \subfigure[]{\includegraphics[width=0.3\textwidth]{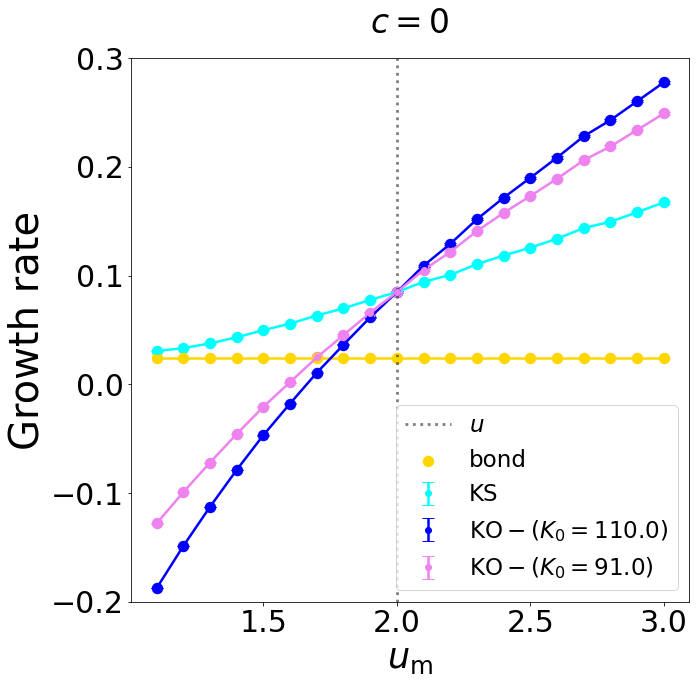}} 
      \subfigure[]{\includegraphics[width=0.3\textwidth]{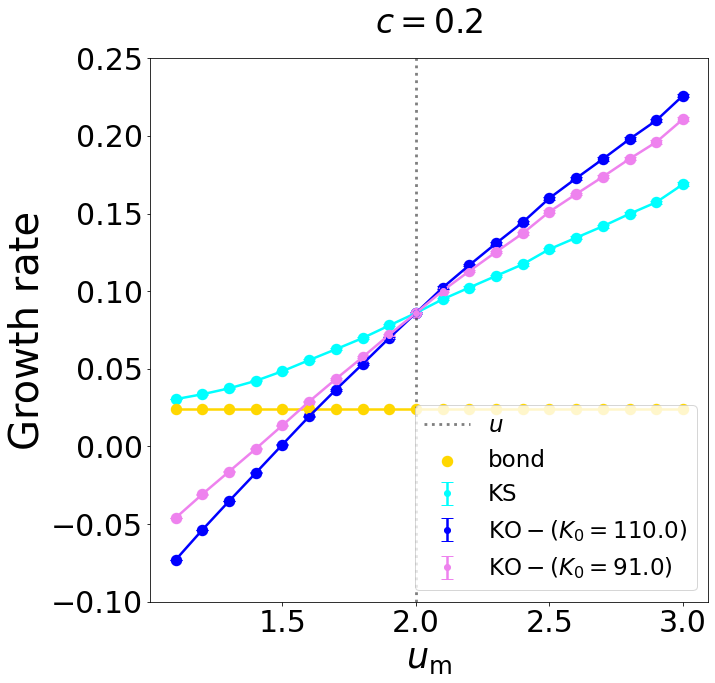}}  
      \subfigure[]{\includegraphics[width=0.3\textwidth]{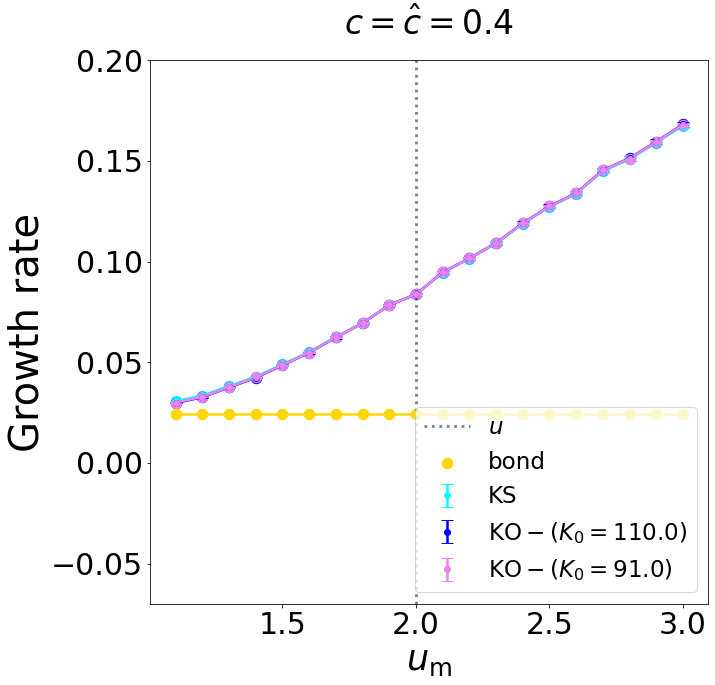}} 
      \subfigure[]{\includegraphics[width=0.3\textwidth]{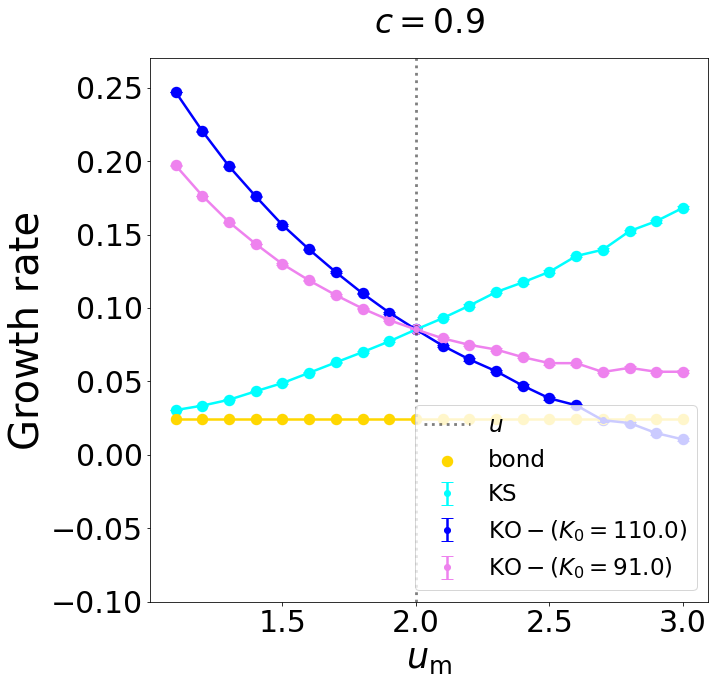}} 
      \subfigure[]{\includegraphics[width=0.3\textwidth]{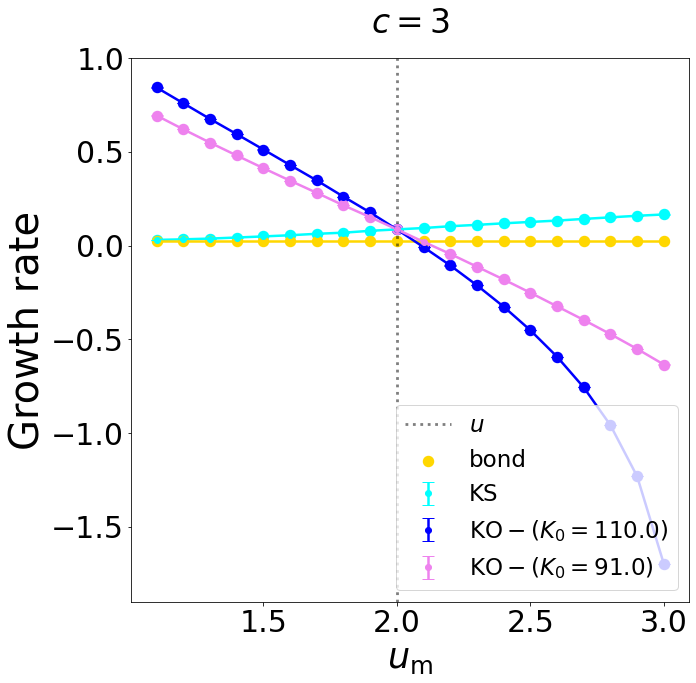}} 
    \caption{\footnotesize Comparison of the growth rate for KS (cyan line) and the  KO (blue and green lines) strategy as a function of $u_\text{m}$. The yellow line represents the bond's growth rate. The other parameters are $u=2$ (gray dotted line), $d_\text{m} = 1/u_\text{m}$, $d=1/u$, $p=0.5$, $R=1.05$, $K_0=110$ and $S_0=100$. The results are averages of $N=500$ simulations with $n=300$ trading rounds each.}
\label{fig: misspecified market V2}
\end{figure}
Figure~\ref{fig: misspecified market V2} shows a numerical estimation for the asymptotically optimal growth rate of KS and KO strategy with out-of the-money option ($K_0=91<S_0$, violet line) and with in-the-money option ($K_0=110>S_0$ blue line) as a function of $u_\text{m}$ when $u=2$ for different values of the hedging strategy parameter $c$. Clearly, when $c$ is chosen so that $g^*=0$, the KO and KS strategies coincide and have the same growth rate under any misspecification. For the chosen parameters, this happens for $c=\hat c=0.4$ (panel (d)). However, the two growth rates differ for a generic $c$; that is, no single KO strategy dominates KS across the whole range of misspecification. KO outperforms KS when $u_\text{m}>u$ ($u_\text{m}<u$) for $c<\hat{c}$ ($c>\hat{c}$), corresponding to a long (short) position on options (see also SI). Since the misspecification is not known, it is not possible to choose a KO strategy that is robust to estimation risk. It is also worth noting that this result is independent of whether the put option is in, at, or out of the money. Finally, we observe that the KS and the KO strategies coincide in the absence of estimation risk, as we also proved theoretically before. 
 
\subsection{Properties of KO solutions and estimation risk}\label{sec:ER}

Here, we better characterize the KO solutions obtained for different $c$ and explain the financial intuition behind the fact that the growth rates of KO and KS differ under misspecification of parameters, i.e. $u_\text{m}$ and $u$ (with the condition $d_\text{m}=1/u_\text{m}$ and $d=1/u$).

\iffalse
\begin{figure}[t]
    \centering
    \subfigure[]{\includegraphics[width=0.3\textwidth]{figures/Comp_KS_KO_cm1.png}}  
      \subfigure[]{\includegraphics[width=0.3\textwidth]{figures/Comp_KS_KO_c0.png}} 
      \subfigure[]{\includegraphics[width=0.3\textwidth]{figures/Comp_KS_KO_c02.png}}  
      \subfigure[]{\includegraphics[width=0.3\textwidth]{figures/Comp_KS_KO_c04.png}} 
      \subfigure[]{\includegraphics[width=0.3\textwidth]{figures/Comp_KS_KO_c09.png}} 
      \subfigure[]{\includegraphics[width=0.3\textwidth]{figures/Comp_KS_KO_c3.png}} 
    \caption{\footnotesize Comparison of the growth rate for KS (cyan line) and the  KO (blue line) strategy as a function of $u_\text{m}$. The yellow line represents the bond's growth rate. The other parameters are $u=2$ (red dotted line), $d_\text{m} = 1/u_\text{m}$, $d=1/u$, $p=0.5$, $R=1.05$, $K_0=110$ and $S_0=100$. The results are averages of $N=500$ simulations with $n=300$ trading rounds each.}
\label{fig: misspecified market V2}
\end{figure}
\fi
Figure \ref{fig: misspecified market V2} shows the growth rate as a function of the mismatch between $u$ and $u_\text{m}$ for different values of $c$. % when $p=1/2$.\footnote{We choose $p=1/2$ following standard market efficiency assumptions and excluding possible bias in the evolution of the stock.} 
When $u_\text{m}>u$ the growth rate of the optimal KS strategy is higher compared to $u=u_\text{m}$. In fact, for $f^*>0$\footnote{The condition for $f^*>0$ is $\E(S_t/S_{t-1})> R$.} , the expected return of the portfolio increases with $u_\text{m}$. 
% for $u_\text{ex-post}>u$, thus the growth rate as well. And vice versa for $u_\text{ex-post}<u$ {\bf FAB: I do not understand: isn't always increasing?}. 
Interestingly, in the misspecified scenarios the growth rate of the optimal KO strategy does not coincide with that of the KS strategy. Moreover, the growth rate of the KO strategy depends on $c$.  
Specifically, it increases 
for $c<\hat{c}$ and outperforms KS when $u_m>u$ 
and it decreases for $c>\hat{c}$ and outperforms KS when $u_m<u$, while KO and KS coincide in the misspecified scenarios when $c=\hat{c}$ (panel (d)).  

Roughly speaking, when $c\ll \hat{c}$, the optimal KO suggests buying more options than stocks, eventually leveraging the investment by short-selling stocks and bonds. Under the misspecification $u_\text{m}>u$, the extra profit can be explained in terms of a mismatch between the option price relative to payoff $\max\{K_{t-1}-dS_t,0\}$ and a realized payoff $K_{t-1}-d_\text{m}S_t>K_{t-1}-dS_t$ in the case of a down movement (and vice versa when $u_\text{m}<u$). On the contrary, when $c\gg \hat{c}$, the optimal KO suggests short-selling put options. Under the misspecification $u_\text{m}<u$, the option is paid higher because $K_{t-1}-dS_t>K_{t-1}-d_\text{m}S_t$, thus justifying an increase of the growth rate, and vice versa when $u_\text{m}>u$.

The following proposition better characterizes the intuition expressed above.

\begin{proposition}\label{prop: relationship between c and g}
 Let be $X\sim\Phi(p)$ an i.i.d. Bernoulli variable associated with the price dynamics, define
    \begin{align*}
        c_u(g) & = -\frac{g\tilde{u}+(1-g)R}{u-R}\\
        c_d(g) & = -\frac{g\tilde{d}+(1-g)R}{d-R}
    \end{align*}
and consider $K_0\in(dS_0,uS_0)$. Let $g_0\in\R$, then it is:
    \begin{itemize}
        \item [1.] For any $g>g_0$, there is an open interval $I_c^{g}$ such that $I_c^{g}\subset (-\infty,c_d(g_0))$ and ${\pi}_{g,c}(X)>0$ a.s. for any $c\in I_c^{g}$;
        \item[2.] For any $g<g_0$, there is an open interval $I_{c}^{g}$ such that $I_c^{g}\subset(c_u(g_0),\infty)$ and ${\pi}_{g,c}(X)>0$ a.s. for any $c\in I_c^{g}$.
    \end{itemize}
\end{proposition}
\begin{proof}
    See Appendix \ref{app:proofs}.
\end{proof}
A corollary of the above proposition is the following.
\begin{corollary}\label{cor: short-selling vs leverage}
    Let be $X\sim\Phi(p)$ an i.i.d. Bernoulli variable associated with the
price dynamics and consider $K_0\in(dS_0,uS_0)$, then it is:
    \begin{itemize}
        \item [1.] If $c<c_d(1)$ then ${\pi}_{g,c}(X)>0$ a.s. for any $g>1$.
        \item [2.] If $c>c_u(0)$ then ${\pi}_{g,c}(X)>0$ a.s. for any $g<0$.
    \end{itemize}
\end{corollary}
Based on Corollary \ref{cor: short-selling vs leverage}, we get a more formal characterization of the optimal KO strategies across different $c$ in Figure~\ref{fig: misspecified market V2}. When $c<c_d(1) = 0.25$, then it is $g>1$, thus the optimal KO strategy considers buying put options on leverage, see panels (a)-(c). When $c>c_u(0) = 1.1$, see panel (f), the optimal KO strategy considers short-selling put options; the other cases, see panels (d)-(e), correspond to $g\in[0,1]$.

\subsection{Convex combination of KO strategies}

To find a robust strategy to any misspecification, we propose to use a convex combination of KO strategies whose growth rate will converge as $n\rightarrow\infty$ to the largest one, which will dominate the other over time. This idea is similar to the one used in Universal Portfolios (\cite{CoverUP}) and in asset allocation strategies (\cite{KanZhou2007,TuZhou2011}). More specifically, we choose two hedging parameters $c_1<\hat{c}$ and $c_2>\hat{c}$ and compute the associated optimal fractions  $g_1^*$ and $g_2^*$ associated with two KO strategies, KO$_1$ and KO$_2$. Then, we invest at each time a fraction $a$ of the wealth in one portfolio and a fraction $1-a$ in the other one, for some $a\in(0,1)$. Let us name KO {\it convex} (KOc) the new strategy. The wealth of KOc at time $n$ results then equal to
\begin{equation}\label{eq:Wnconvex}
W_n^{\mbox{KOc}} = aW_n^\text{(1)}+(1-a)W_n^\text{(2)}.
\end{equation}
%For a chosen $a\in(0,1)$ we build a convex combination of the two KO strategies and we indicate with ${W}_n^{\mbox{KOc}}$ the wealth of KO {\it convex} (KOc) strategy at time $n$. 
%\noindent In the SI we prove the following.

%For the generic growth rate at time $n$ associated with the KOc strategy, it is
%\begin{align}\label{eq: Jensen bound}
%    \frac{1}{n}\log(\widetilde{W}_n^{\mbox{KOc}}) &= \frac{1}{n}\log(a\widetilde{W}_n^\text{(1)}+(1-a)\widetilde{W}_n^\text{(2)})\nonumber\\&\geq a\frac{\log(\widetilde{W}_n^\text{(1)})}{n}+(1-a)\frac{\log(\widetilde{W}_n^\text{(2)})}{n}.
%\end{align}
%because of Jensen's inequality. However, the lower bound for the growth rate does not allow for inferring the asymptotic behavior for the expected growth rate $\lim_{n\rightarrow\infty} \frac{1}{n}\log(\widetilde{W}_n^{\mbox{KOc}})$, in particular to consider a comparison with the classical Kelly solution $G({f^*},\Phi)$ under estimation risk. We characterize the asymptotic behavior of KOc in terms of growth rate in the following theorem.

\begin{theorem}\label{thrm: thorem convex strategy}
Let $(\Omega, \{\mathcal{F}_t\}_{t=0}^n,\PP)$ be the probability space associated with the binomial tree market, and $X\sim\Phi(p)$ an i.i.d. Bernoulli variable  associated with the price dynamics. Then, the asymptotic exponential growth rate of KOc is %{\bf FAB: Shouldn't we divide ${W}_n^{\mbox{KOc}}$ by $W_0$?}
\begin{equation}\label{agrKOc}
        \lim_{n\rightarrow\infty} \frac{1}{n}\log\left(\frac{{W}_n^{\mbox{KOc}}}{W_0}\right) = \max\big\{\E[\log{\pi}_{g_1^*,c_1}(X)],\E[\log{\pi}_{g_2^*,c_2}(X)]\big\}\:\:\mbox{a.s.}
    \end{equation}
% \begin{equation}\label{agrKOc}
%         \lim_{n\rightarrow\infty} \frac{1}{n}\log\big(\widetilde{W}_n^{\mbox{KOc}}\big) \xrightarrow{a.s.} \max\big\{\E[\log\tilde{\pi}_{g_1^*,c_1}(X)],\E[\log\tilde{\pi}_{g_2^*,c_2}(X)]\big\}\:\:\mbox{as }n\rightarrow\infty
%     \end{equation}
\end{theorem}
\begin{proof}
    See appendix \ref{app:proofs}.
\end{proof}

\begin{figure}
    \centering
    \subfigure[]
{\includegraphics[width=0.44\textwidth]{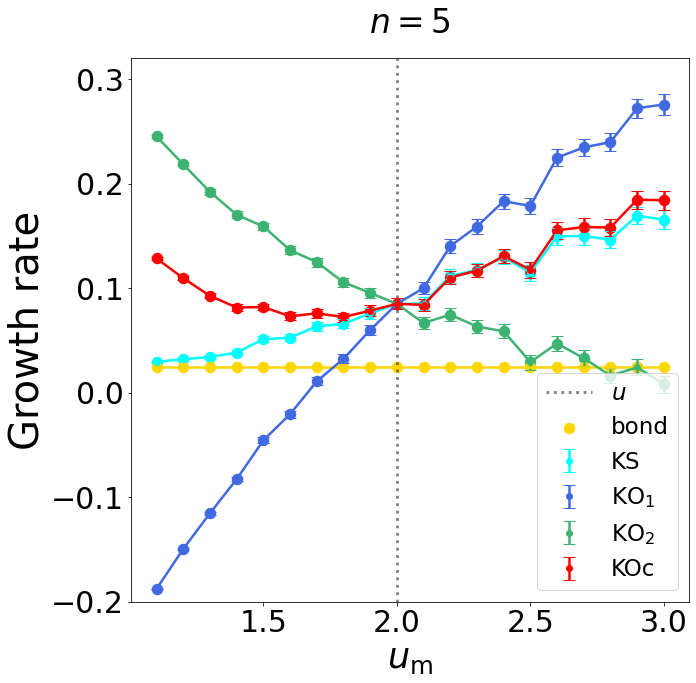}} 
    \subfigure[]
{\includegraphics[width=0.44\textwidth]{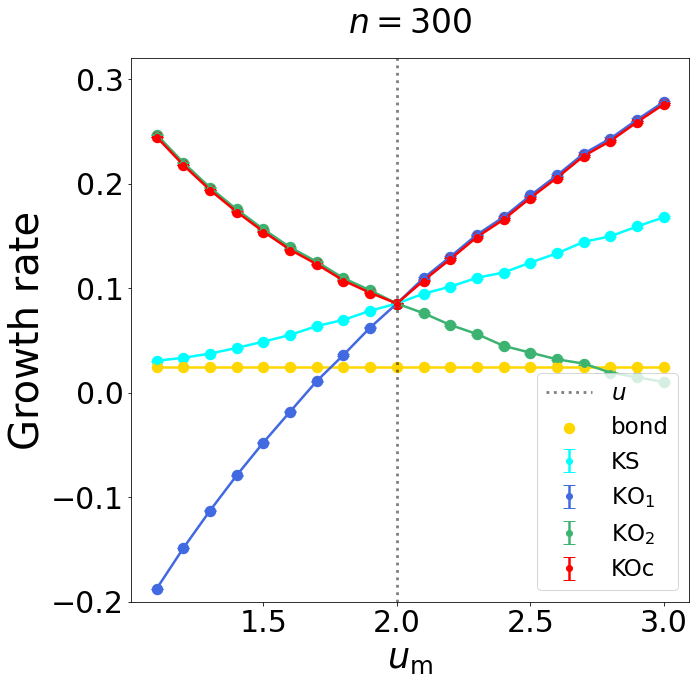}} 
    \caption{\footnotesize Comparison of optimal Kelly strategies in the presence of estimation risk. The strike price at time $0$ is $K_0=110$\euro~, $(c_1,c_2) = (0,0.9)$ and the probability of an upward move is $p=0.5$. The analysis is based on $N=500$ simulations over different trading periods. The parameters are the same as in Figure \ref{fig: misspecified market V2}. The KOc strategy is obtained with $a=1/2$.}
    \label{fig: Comp KO,KS,KOVKOD,KOVKS}
\end{figure}

Notice that this is an asymptotic result and might not hold for finite time. Figure~\ref{fig: Comp KO,KS,KOVKOD,KOVKS} shows the simulated growth rate for $n=5$ (left) and $n=300$ (right) trading periods. For a small time horizon, the KOc strategy does not outperform KO$_1$ and KO$_2$ under any misspecification of the parameters; however, in the long run, the growth rate of KOc coincides with the best strategy under a specific misspecification. As such, the asymptotic growth rate of the KOc portfolio is always larger than (or equal to, when $u=u_\text{m}$) the one achieved by the KS strategy, showing that the estimation risk has been fully eliminated.

\begin{figure}[h]
    \centering
    \subfigure[]{\includegraphics[width=0.24\textwidth]{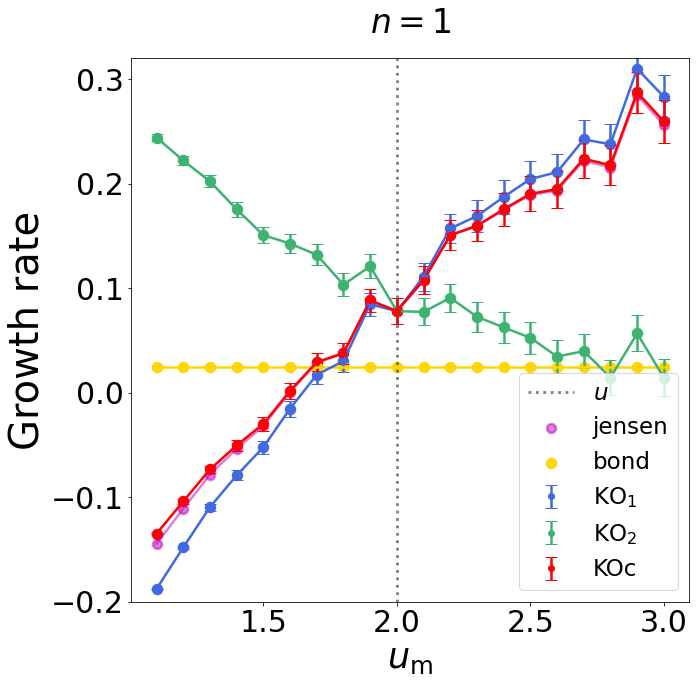}}  
      \subfigure[]{\includegraphics[width=0.24\textwidth]{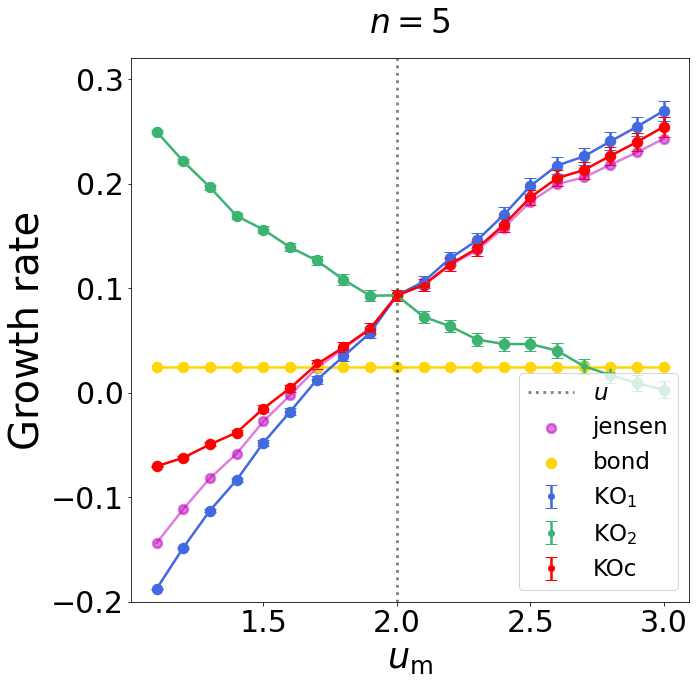}} 
      \subfigure[]{\includegraphics[width=0.24\textwidth]{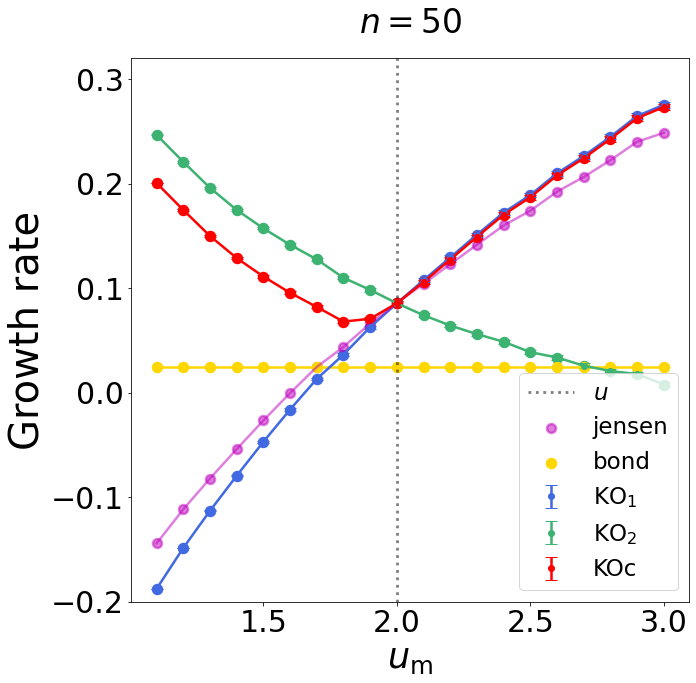}}  
      \subfigure[]{\includegraphics[width=0.24\textwidth]{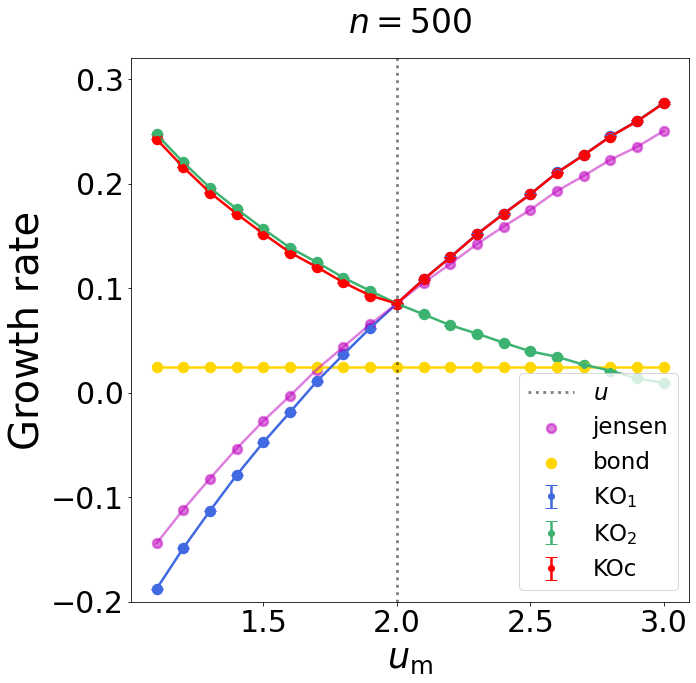}} 
       \subfigure[]{\includegraphics[width=0.24\textwidth]{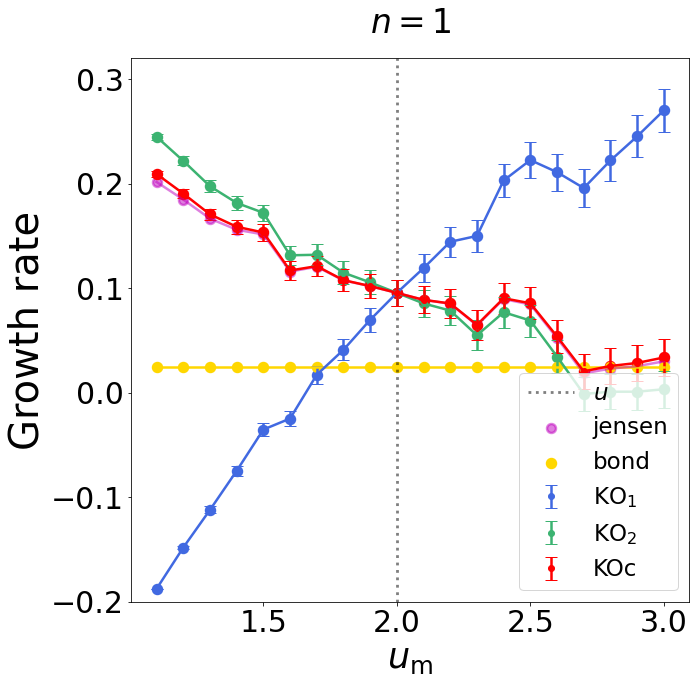}}  
      \subfigure[]{\includegraphics[width=0.24\textwidth]{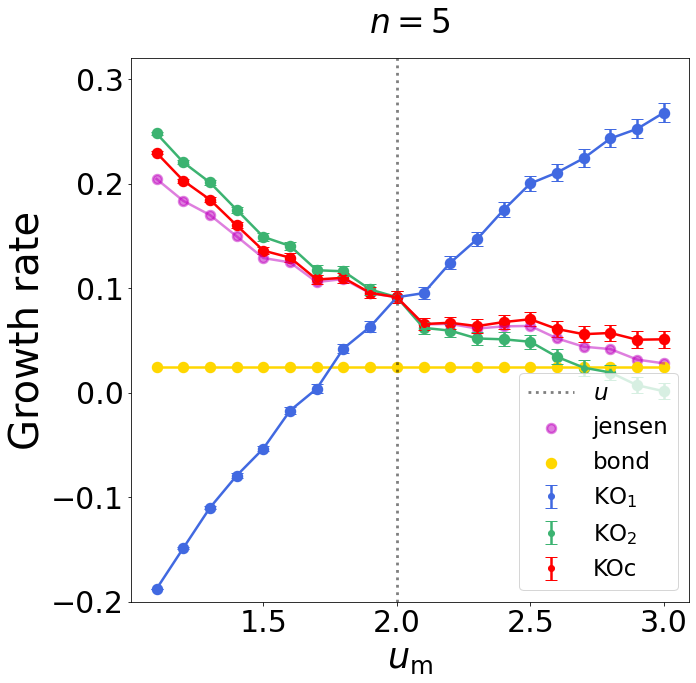}} 
      \subfigure[]{\includegraphics[width=0.24\textwidth]{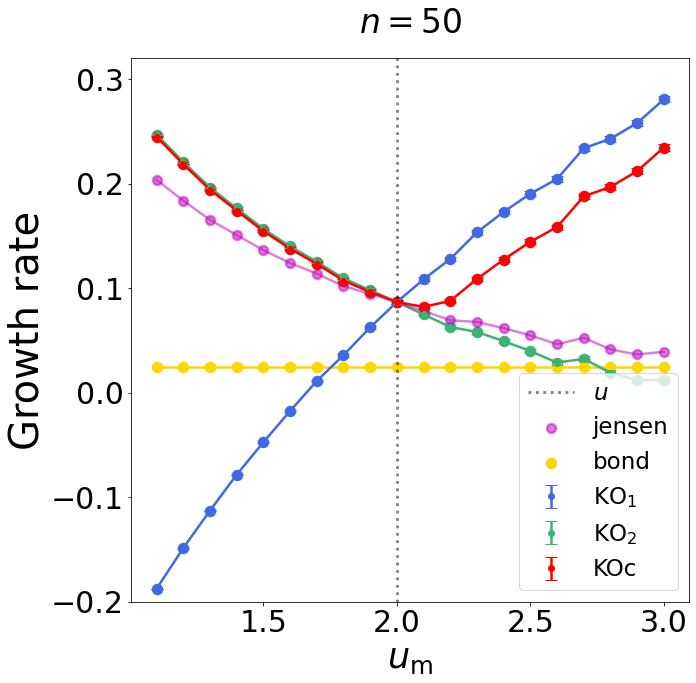}}  
      \subfigure[]{\includegraphics[width=0.24\textwidth]{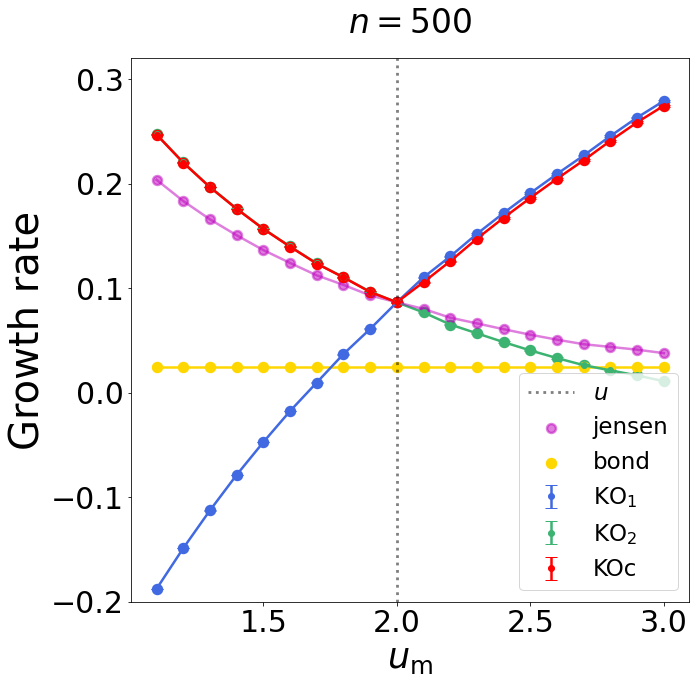}} 
    \caption{\footnotesize Comparison of the finite time growth rate in a misspecified setting. The parameters are set to $u=2$, $d=1/u$, $K_0=110$, $(c_1,c_2) = (0,0.9)$ and $d_\text{m} = 1/u_\text{m}$. The results are obtained by averaging over $N=500$ simulations and different panels refer to different investment time $n$ and mixing parameter $a$. {\it Top  panels}: $a=0.9$. {\it Bottom panels}: $a=0.1$.}
    \label{fig: jensen a=0.9,0.1}
\end{figure}
More detailed simulation results are shown in Figure~\ref{fig: jensen a=0.9,0.1}, where a comparative analysis is presented to characterize the growth rate of the different Kelly strategies for finite investment periods $n$ in the presence of misspecified parameters. While for small $n$, KOc is not the optimal strategy in the whole range of misspecification, it becomes the best strategy for any $u_{m}$ for large $n$. The figure also shows a purple line (termed as ``Jensen") showing the dynamics of
\begin{equation}\label{eq: jensen}
    a\frac{\log W_n^{(1)}}{n}+(1-a)\frac{\log W_n^{(2)}}{n}.
\end{equation} 
Because of Jensen's inequality, the growth rate of the KOc strategy is expected to be always above the one in Eq.(\ref{eq: jensen}), as observed numerically.% \textcolor{blue}{which shows that investing according to a convex combination of the strategies yields better performance than investing fully in each strategy and only afterward forming a convex combination.}

From a practitioner's point of view, the implementation of the KOc strategy requires the creation of two KO portfolios, namely two different {\it portfolio insurance strategies}\footnote{Portfolio insurance is the technical word used among practitioners to refer to the hedging strategy defined within the Kelly framework in Section~\ref{seq:kelly}, see for example \cite{benninga1985optimality,ho2010portfolio}.} with hedging parameters $c_1$ and $c_2$. One might ask whether the two strategies can be replaced by two replicating portfolios of stock and bond only, thus avoiding option trading (in particular, short selling of options), which could be prevented for some investors. However, in line with Proposition~\ref{prop:1}, the only optimal Kelly strategy using stock and bond only is the KS solution. As such, the KOc would coincide with KS, which, as we have seen, is sensitive to estimation risk. Therefore, the use of options remains essential, since the two $\mbox{KO}_1$ and $\mbox{KO}_2$ strategies are needed to cover from misspecification of parameters in both directions, exploiting the mismatch between option prices and actual market realizations.

% This involves selecting $c_1,c_2$ such that $c_1\ll\hat{c}\ll c_2$. Determining $\hat{c}$, is straightforward, as it corresponds to the value of $c$ such that $g^*=0$. The next step is to use options with strike prices $K_t$ that satisfy Eq.(\ref{eq: condition of stationarity}). From a practitioner’s perspective, one might attempt to avoid the difficulty of selecting the option strike price by replicating the two KO portfolios directly. However, this approach would again lead to the KS strategy. As shown in Proposition \ref{prop: replicating KO for any c}, the replicating portfolio of any KO strategy coincides with the KS strategy, irrespective of the value of $c$ (see proposition \ref{prop: replicating KO for any c}). Consequently, this approach does not mitigate estimation risk. Therefore, the use of options remains essential.

\section{Conclusions}\label{sec:concl}

The Kelly criterion, introduced by \cite{kelly}, revolutionized the fields of gambling and portfolio optimization by providing a robust framework for maximizing long-term wealth growth while controlling market risk. However, high sensitivity to estimation risk has long been noted as one of the practical limitations of the investment approach. This paper proposes a solution to this problem for a binomial tree market by integrating option trading with log-optimal portfolios to mitigate estimation risk within the Kelly framework. A proper convex combination of Kelly with Options (KO) strategies is proved to be asymptotically robust to any parameter misspecification. 

In continuous time, Kelly criterion, namely  maximizing the expected logarithmic utility of wealth, leads to the Merton's portfolio problem. For example, in a Black-Scholes market with a stock and a bond, the solution for the growth-optimal (Kelly) policy provides a constant fraction to invest in the stock and the remaining in the bond, see, e.g., \cite{merton1969lifetime}, similarly to the standard KS solution. Uncertainty in the form of incomplete information about the price dynamics has been studied for Merton's portfolio problem, considering for example unknown drift (\cite{lakner1998optimal}), unobserved market regimes (\cite{sass2004optimizing}), or latent jump processes (\cite{callegaro2006portfolio}), typically showing that the standard Merton solution is modified by using filtered values in the place of original parameters. Whereas including European option trading has been seen as {\it redundant} in terms of optimal growth in complete markets since \cite{merton1969lifetime} (consistently with Proposition \ref{prop:1}), \cite{romano1997contingent} and works hereafter have shown however that adding options can be useful for market-completion in the case of stochastic volatility models and for handling variance risk robustly. Uncertainty and estimation risks remain an open point within this last context, and the generalization of the KO approach to continuous time appears to be the natural outlook of the present work.

These advancements underscore the importance of adaptive and hybrid strategies in achieving optimal portfolio performance. By combining theoretical insights with practical considerations, the extended Kelly strategies provide a robust foundation for navigating complex financial markets while addressing inherent uncertainties.

\section*{Acknowledgments}
Authors thanks Stefano Marmi for the useful discussions. Piero Mazzarisi acknowledges funding from the Italian Ministry of University and Research under the PRIN project {\it Realized Random Graphs: A New Econometric Methodology for the Inference of Dynamic Networks} (grant agreement n. 2022MRSYB7).

\appendix
\section{Proofs}\label{app:proofs}
This section contains the proofs of propositions and theorems presented in the paper.
%\iffalse
\subsection{Preliminaries}
First, we prove the following lemma.
\begin{lemma}\label{first lemma}
    Let $K_0\in(dS_0,uS_0)$ then it holds:
    \begin{equation}\label{eq: objective in lemma}
        \frac{R-\tilde{u}}{R-u}   = \frac{R-\tilde{d}}{R-d}>0
    \end{equation}
where 
\begin{equation}
    \tilde{u} = u\frac{S_0}{P_0}-R\frac{S_0}{P_0}\quad\text{and}\quad\tilde{d} = \frac{K_0}{P_0}-R\frac{S_0}{P_0}.
\end{equation}
\end{lemma}
\begin{proof}
Since $K_0\in(dS_0,uS_0)$ the price of the put option is
\begin{equation}
    P_0 = \frac{1}{R}\frac{u-R}{u-d}(K_0-dS_0).
\end{equation}
We first show that the equality holds
    \begin{align}
        \frac{R-\Tilde{u}}{R-\Tilde{d}}   = \frac{R-u}{R-d}\Leftrightarrow
    \frac{P_0R-uS_0+RS_0}{P_0R-K_0+RS_0}  = \frac{R-u}{R-d}\Leftrightarrow\nonumber\\[15pt]
     P_0R(R-d)+S_0(R-u)(R-d)  = P_0R(R-u)+(S_0R-K_0)(R-u)\Leftrightarrow\label{eq:terms with P0}\\[10pt]
     (u-R)(K_0-dS_0)=(R-u)(dS_0-K_0)\label{eq:final equation for wealths equality}
    \end{align}
hence the thesis.

For the positivity, due to the equality above, it is enough to show
\begin{equation}\label{ineq: positivity}
    \frac{R-\tilde{u}}{R-u}>0.
\end{equation}
Since from the no-arbitrage condition it is $R-u<0$, to show the inequality (\ref{ineq: positivity}), it is needed to show
\begin{align}
    R-\tilde{u}<0\Leftrightarrow
    (u-R)\frac{S_0}{P_0}>R\Leftrightarrow
    (u-R)S_0>RP_0\Leftrightarrow\nonumber\\
    (u-R)S_0>\frac{u-R}{u-d}(K_0-dS_0)\Leftrightarrow
    S_0(u-d)>K_0-dS_0\Leftrightarrow\nonumber\\\label{eq: last inequality}
    uS_0>K_0.
\end{align}
Inequality (\ref{eq: last inequality}) holds by assumption, thus inequality (\ref{ineq: positivity}) is also true.
\end{proof}

\subsection{Proof of Proposition \ref{prop:1} of the paper}
In order to prove the almost sure equality, we need to show that the relative payoff for KS equals the relative payoff for KO strategy both when $X=u$ and when $X=d$. Since the calculations are similar, we focus only on the case $X=u$.
\begin{align}
    {\pi}_{g^*,c}(u) & = {\pi}_{0,f^*}(u)\Leftrightarrow\nonumber\\
    g^*(\Tilde{u}-R)+R+c(u-R) & = f^*(u-R)+R\Leftrightarrow\nonumber\\
 g^*(\Tilde{u}-R) & = (f^*-c)(u-R)\Leftrightarrow\nonumber\\[10pt]
 \frac{\frac{p(R-\Tilde{u})(R+c(d-R))+(1-p)(R-\Tilde{d})(R+c(u-R))}{(\Tilde{d}-R)(\Tilde{u}-R)}}{\frac{p(R-u)R+(1-p)(R-d)R-c(R-u)(R-d)}{(u-R)(d-R)}}& = \frac{u-R}{\Tilde{u}-R}\Leftrightarrow\nonumber\\[10pt]
 \frac{p(R-\Tilde{u})(R+c(d-R))+(1-p)(R-\Tilde{d})(R+c(u-R))}{p(R-u)R+(1-p)(R-d)R-c(R-u)(R-d)}& = \frac{R-\Tilde{d}}{R-d}\Leftrightarrow\nonumber\\[10pt]
 \frac{p\frac{R-\Tilde{u}}{R-\Tilde{d}}R-p\frac{R-\Tilde{u}}{R-\Tilde{d}}c(R-d)+(1-p)R-(1-p)c(R-u)}{p\frac{R-u}{R-d}R+(1-p)R-c(R-u)} & = 1
\end{align}
This last equality holds since the numerator equals the denominator of the fraction in the left hand side. Using the lemma \ref{first lemma}, the following part in the numerator vanishes
\begin{align*}
    -p\frac{R-\Tilde{u}}{R-\Tilde{d}}c(R-d)+pc(R-u) = 0\Leftrightarrow
    -p\frac{R-u}{R-d}c(R-d)+pc(R-u) = 0.
\end{align*}

\subsection{Proof of Theorem \ref{thrm: thorem convex strategy}}
\begin{lemma}[Log-Sum Inequality \cite{Cover_Thomas_Info_Theory} -- theorem 17.1.2]\label{lem: Log-Sum ineq}
For positive numbers $a_1,a_2,\cdots,a_n$ and $b_1,b_2,\cdots,b_n$,
\begin{equation*}
    \sum_{i=1}^n a_i\log\frac{a_i}{b_i}\geq \sum_{i=1}^n a_i\log\frac{\sum_{i=1}^n a_i}{\sum_{i=1}^n b_i}.
\end{equation*}
\end{lemma}
% \begin{lemma}[Log-Sum-Exp Inequality]
% Let $x_i\in\R$ for $i=1,\cdots,n$, then
% \begin{equation*}
%    \max_i x_i\leq \log\sum_{i=1}^n e^{x_i}\leq \max_i x_i+\log n.
% \end{equation*}
% \end{lemma}
% \begin{proof}
%     For the upper bound we set $a_i=e^{x_i}$ and $b_i=1$ for $i=1,\cdots,n$ in the Log-Sum inequality (see lemma \ref{lem: Log-Sum ineq}),
%     \begin{align*}
%         \sum_{i=1}^n e^{x_i}\log\frac{e^{x_i}}{1}&\geq \Big(\sum_{i=1}^ne^{x_i}\Big)\log\frac{\sum_{i=1}^ne^{x_i}}{\sum_{i=1}^n1}\Leftrightarrow\\
%          \sum_{i=1}^n e^{x_i}x_i&\geq\sum_{i=1}^ne^{x_i}\Big(\log\sum_{i=1}^ne^{x_i}-\log n\Big)\Leftrightarrow\\
%          \frac{\sum_{i=1}^n e^{x_i}x_i}{\sum_{i=1}^n e^{x_i}}&\geq\log\sum_{i=1}^ne^{x_i}-\log n\Leftrightarrow\\
%          \log\sum_{i=1}^ne^{x_i}&\leq  \frac{\sum_{i=1}^n e^{x_i}x_i}{\sum_{i=1}^n e^{x_i}}+\log n\leq\max_i x_i+\log n
%     \end{align*}
%     for the last inequality we observe that
%     \begin{equation*}
%         \frac{\sum_{i=1}^n e^{x_i}x_i}{\sum_{i=1}^n e^{x_i}}\leq \frac{\max_i x_i\sum_{i=1}^n e^{x_i}}{\sum_{i=1}^n e^{x_i}}.
%     \end{equation*}
%     For the lower bound we use the inequality
%     \begin{equation*}
%         \sum_{i=1}^ne^{x_i}\geq e^{\max_i x_i}\Leftrightarrow\log\sum_{i=1}^ne^{x_i}\geq\max_i x_i.
%     \end{equation*}
% \end{proof}
\begin{lemma}[Convex-Log-Sum-Exp Inequality]\label{lem: Convex-Log-Sum-Exp}
Let $x_i\in\R$ for $i=1,\cdots,n$ and $\lambda_i\in(0,1)$ for $i=1,\cdots,n$ such that $\sum_{i=1}^n\lambda_i=1$, then  
\begin{equation*}
     \max_i x_i+ \log\min_i\lambda_i\leq \log\sum_{i=1}^n \lambda_ie^{x_i}\leq \max_i x_i.
\end{equation*}
\end{lemma}
\begin{proof}
For the lower bound we observe that
\begin{equation*}
    \sum_{i=1}^n\lambda_ie^{x_i} = e^{M}\sum_{i=1}^n\lambda_ie^{x_i-M}\Leftrightarrow\log\sum_{i=1}^n\lambda_ie^{x_i} = M+\log\sum_{i=1}^n\lambda_ie^{x_i-M}
\end{equation*}
where $M=\max_ix_i$. Then it is
\begin{equation*}
    \min_i\lambda_i\leq\sum_{i=1}^n\lambda_ie^{x_i-M}
\end{equation*}
since, from the definition of $M$, there is $i\in\{1,\cdots,n\}$ such that $e^{x_i-M}=1$ and for every $i\in\{1,\cdots,n\}$ it is $e^{x_i-M}>0$. Hence
\begin{equation*}
    \log\sum_{i=1}^n\lambda_ie^{x_i}\geq \max_i x_i+\log\min_i\lambda_i.
\end{equation*}
For the upper bound we set $a_i=\lambda_ie^{x_i}$ and $b_i=\lambda_i$ for $i=1,\cdots,n$ in the Log-Sum inequality (see lemma \ref{lem: Log-Sum ineq}),
    \begin{align*}
        \sum_{i=1}^n\lambda_i e^{x_i}\log\frac{\lambda_ie^{x_i}}{\lambda_i}&\geq \sum_{i=1}^n\lambda_ie^{x_i}\log\frac{\sum_{i=1}^n\lambda_ie^{x_i}}{\sum_{i=1}^n\lambda_i}\Leftrightarrow\\
        \sum_{i=1}^n\lambda_i e^{x_i}x_i&\geq \sum_{i=1}^n\lambda_ie^{x_i}\log\sum_{i=1}^n\lambda_ie^{x_i}\Leftrightarrow\\
       \log\sum_{i=1}^n\lambda_ie^{x_i}&\leq \frac{\sum_{i=1}^n\lambda_i e^{x_i}x_i}{\sum_{i=1}^n\lambda_ie^{x_i}}\leq\max_ix_i
    \end{align*}
    for the last inequality we observe that
    \begin{equation*}
        \frac{\sum_{i=1}^n \lambda_ie^{x_i}x_i}{\sum_{i=1}^n \lambda_ie^{x_i}}\leq \frac{\max_i x_i\sum_{i=1}^n \lambda_ie^{x_i}}{\sum_{i=1}^n \lambda_ie^{x_i}}.
    \end{equation*}
    As a final remark of this lemma, we observe that $\log\min_i\lambda_i<0$, since $\min_i\lambda_i<1$, thus for the upper and lower bound it is
    \begin{equation*}
        \max_i x_i+\log\min_i\lambda_i<\max_i x_i.
    \end{equation*}
\end{proof}
We proceed now to the proof of the theorem.
\begin{proof}
Without loss of generality, we assume that $W_0=1$, thus it is
\begin{align}
    \lim_{n\rightarrow\infty} \frac{1}{n}\log\big({W}_n^\text{KOc}\big) & = \lim_{n\rightarrow\infty} \log\big({W}_n^\text{KOc}\big)^{1/n}\nonumber\\
    & = \lim_{n\rightarrow\infty} \log\big(a{W}_n^{(1)}+(1-a){W}_n^{(2)}\big)^{1/n}\nonumber\\
    & = \lim_{n\rightarrow\infty} \log\big(ae^{\log{W}_n^{(1)}}+(1-a)e^{\log{W}_n^{(2)}}\big)^{1/n}\nonumber\\
    & = \lim_{n\rightarrow\infty} \log\big(ae^{\sum_{t=1}^n \log {\pi}_{g_1^*,c_1}(X_t)}+(1-a)e^{\sum_{t=1}^n \log {\pi}_{g_2^*,c_2}(X_t)}\big)^{1/n}.\label{eq: limit of V}
\end{align}
Let us define
\begin{equation*}
    Y_n = ae^{\sum_{t=1}^n\log{\pi}_{g_1^*,c_1}(X_t)}+(1-a)e^{\sum_{t=1}^n\log{\pi}_{g_2^*,c_2}(X_t)},
\end{equation*}
and by using the Convex-Log-Sum-Exp inequality (see Lemma \ref{lem: Convex-Log-Sum-Exp}), in which we divide by $n$, it is
\begin{align*}
    \frac{1}{n}\max\{\sum_{t=1}^n \log {\pi}_{g_1^*,c_1}(X_t),\sum_{t=1}^n \log {\pi}_{g_2^*,c_2}(X_t)\}+\frac{\log \min\{a,1-a\}}{n}\\\leq\log Y_n^{1/n}\leq\frac{1}{n}\max\{\sum_{t=1}^n \log {\pi}_{g_1^*,c_1}(X_t),\sum_{t=1}^n \log {\pi}_{g_2^*,c_2}(X_t)\}.
\end{align*}
Finally, from Eq.~(4) of the paper and using the {\it Squeeze theorem}, it follows
\begin{equation}
    \lim_{n\rightarrow\infty}\frac{1}{n}\log Y_n = \max\big\{\E[\log {\pi}_{g_1^*,c_1}(X)],\E[\log {\pi}_{g_1^*,c_1}(X)]\big\}~~\text{a.s.}.
    \end{equation}
\end{proof}
As a final remark, we note that in this proof the KOc strategy is asymptotically independent of the choice of $a$, a result also confirmed numerically in Figure \ref{fig: jensen a=0.9,0.1} of the main paper(see panels (d) and (h) compared to (a) and (e)). %For small $n$, however, since $\log\min\{a,1-a\}<0$, the larger the value of $\min\{a,1-a\}$, the closer the KOc strategy lies to the maximum of the two strategies. This effect is visible in Figure 3(a) of the paper, in contrast to figures 2(b) and 2(f), where $n=5$: when $a=1/2$, KOc is closer to the maximum of $\max\{\text{KO}_1,\text{KO}_2\}$ than when $a<1/2$.
% \begin{equation}
%     \lim_{n\rightarrow\infty}\frac{1}{n}\log Y_n \xrightarrow{a.s.} \max\big\{\E[\log \Tilde{\pi}_{g_1^*,c_1}(X)],\E[\log \Tilde{\pi}_{g_1^*,c_1}(X)]\big\}\:\:\mbox{as }n\rightarrow\infty.
%     \end{equation}

\subsection{Proof of Proposition \ref{prop: relationship between c and g}}
To prove the proposition, we first prove the following lemma
\begin{lemma}\label{second lemma}
    Let the functions $c_u,c_d:\R\rightarrow\R$, defined as
    \begin{align*}
        c_u(g) & = -\frac{g\tilde{u}+(1-g)R}{u-R}\\
        c_d(g) & = -\frac{g\tilde{d}+(1-g)R}{d-R}
    \end{align*}
and assume $K_0\in(dS_0,uS_0)$ then the following holds:
\begin{itemize}
    \item[1.] $c_u(g)<c_d(g)$, for any $g\in\R$.
    \item[2.] For any $g_1,g_2\in\R$ with $g_1<g_2$ holds
\begin{align*}
   c_u(g_1)&>c_u(g_2)~~~\text{and}\\
   c_d(g_1)&>c_d(g_2).
\end{align*}
\end{itemize}
\end{lemma}
\begin{proof}
\begin{itemize}
    \item[1.] It is 
    \begin{align}
        c_u(g) = -g\frac{\Tilde{u}-R}{u-R}-\frac{R}{u-R}
        = -g\frac{\Tilde{d}-R}{d-R}-\frac{R}{u-R}
        < -g\frac{\Tilde{d}-R}{d-R}-\frac{R}{d-R} = c_d(g)
    \end{align}
    The second equality follows from lemma~\ref{first lemma}, while the inequality follows from the fact that
    \begin{equation}
        -\frac{R}{u-R}<0<-\frac{R}{d-R}.
    \end{equation}
    \item[2.] The functions $c_u(g),c_d(g)$ are both decreasing due to lemma \ref{first lemma}:
    \begin{align}
        c_u^\prime(g) & = -\frac{\tilde{u}-R}{u-R}<0\\
        c_d^\prime(g) & = -\frac{\tilde{d}-R}{d-R}<0
    \end{align}
    thus for any $g_1,g_2\in[0,1]$ with $g_1<g_2$ holds that $c_u(g_1)>c_u(g_2)$ and $c_d(g_1)>c_d(g_2)$.
\end{itemize}
\end{proof}
From lemma \ref{second lemma} holds $c_u(g)<c_d(g)$ for any $g\in\R$ and also that
    \begin{itemize}
        \item [1.] For any $g>g_0$ it is $c_d(g)<c_d(g_0)$ and $c_u(g)<c_u(g_0)$. Thus, by setting $I_c^{g} = (c_u(g),c_d(g))$ it is true that $I_c^{g}\subset (-\infty,c_d(g_0))$. 
        \item[2.] For any $g<g_0$ it is $c_d(g)>c_d(g_0)$ and $c_u(g)>c_u(g_0)$. Thus, by setting $I_c^{g} = (c_u(g),c_d(g))$ it is true that $I_c^{g}\subset (c_u(g_0),\infty)$. 
    \end{itemize}
    Moreover, for any $c\in I_c^{g}$ holds that ${\pi}_{g,c}(X)>0$ a.s.. To prove this last statement, we define
    \begin{align*}
        g_u(c) & = g\Tilde{u}+(1-g)R+c(u-R),\\
        g_d(c) & = g\Tilde{d}+(1-g)R+c(d-R)
    \end{align*}
    which corresponds to ${\pi}_{g,c}(u)$ and ${\pi}_{g,c}(d)$ respectively. Then $g_u$ is an increasing function while $g_d$ is a decreasing function due to the no-arbitrage condition $d<R.<u$.
    Thus for $g\in\R$, $g_u(c)>0$ for any $c>c_u(g)$ and $g_d(c)>0$ for any $c<c_d(g)$. Overall and due to the relation: $c_u(g)<c_d(g)$ (from lemma \ref{second lemma}), ${\pi}_{g,c}(X)>0$ a.s. for any $c\in I_c^{g} = (c_u(g),c_d(g))$.

%===================================
%===================================
% Refs
\bibliographystyle{abbrvnat}
\begingroup
\bibliography{biblio2}
\endgroup

\section*{Email Addresses}
\noindent
\noindent\texttt{fabrizio.lillo@sns.it}\\
\noindent\texttt{piero.mazzarisi@unisi.it}\\
\noindent\texttt{ioannayvonni.tsaknaki@sns.it}\text{~(Corresponding author)}

\end{document}